\documentclass[
nomarks,final
]{dmtcs-episciences}

\usepackage{amssymb,amsmath}
\usepackage{amsmath}
\usepackage{dsfont}
\usepackage{amsfonts}
\usepackage{verbatim}
\usepackage[square,numbers]{natbib}
\usepackage[utf8]{inputenc}
\usepackage{subfigure}

\author{Hamid Kameli\affiliationmark{1}\thanks{Partially supported by Shahid Beheshti University, G.C.}}
\title{Non-adaptive Group Testing on Graphs}

\affiliation{
Department of Mathematical Sciences,
Shahid Beheshti University, G.C. Iran }
\received{2016-1-27}
\revised{2017-5-3, 2018-1-6}
\accepted{2018-1-25}

\newtheorem{precor}{{\bf Corollary}}

\newtheorem{precon}{{\bf Conjecture}}

\newtheorem{prealphcon}{{\bf Conjecture}}

\newtheorem{predefin}{{\bf Definition}}

\newtheorem{preexm}{{\bf Example}}

\newtheorem{preappl}{{\bf Application}}

\newtheorem{prelem}{{\bf Lemma}}

\newenvironment{lem}{\begin{prelem}{\hspace{-0.5
               em}{\bf.\ }}}{\end{prelem}}
\newtheorem{preproof}{{\bf Proof.\ }}

\newenvironment{proof}[1]{\begin{preproof}{\rm
               #1}\hfill{$\blacksquare$}}{\end{preproof}}
\newtheorem{prethm}{{\bf Theorem}}

\newenvironment{thm}{\begin{prethm}{\hspace{-0.5
               em}{\bf.\ }}}{\end{prethm}}
\newtheorem{prealphthm}{{\bf Theorem}}

\newtheorem{prealphlem}{{\bf Lemma}}

\newenvironment{alphlem}{\begin{prealphlem}{\hspace{-0.5
               em}{\bf.\ }}}{\end{prealphlem}}
\newtheorem{prepro}{{\bf Proposition}}

\newtheorem{preprb}{{\bf Problem}}

\newtheorem{prerem}{{\bf Remark}}

\newtheorem{preapp}{{\bf Application}}

\newtheorem{prequ}{{\bf Question}}

\newtheorem{preclaim}{{\bf Claim}}

\def\conct[#1,#2]{\mbox {${#1} \leftrightarrow {#2}$}}
\def\dconct[#1,#2]{\mbox {${#1} \rightarrow {#2}$}}
\def\deg[#1,#2]{\mbox {$d_{_{#1}}(#2)$}}
\def\mindeg[#1]{\mbox {$\delta_{_{#1}}$}}
\def\maxdeg[#1]{\mbox {$\Delta_{_{#1}}$}}
\def\outdeg[#1,#2]{\mbox {$d_{_{#1}}^{^+}(#2)$}}
\def\minoutdeg[#1]{\mbox {$\delta_{_{#1}}^{^+}$}}
\def\maxoutdeg[#1]{\mbox {$\Delta_{_{#1}}^{^+}$}}
\def\indeg[#1,#2]{\mbox {$d_{_{#1}}^{^-}(#2)$}}
\def\minindeg[#1]{\mbox {$\delta_{_{#1}}^{^-}$}}
\def\maxindeg[#1]{\mbox {$\Delta_{_{#1}}^{^-}$}}

\def\dre[#1,#2,#3]{\mbox {${\cal E}^{^{#3}}(#1,#2)$}}
\def\var[#1,#2]{\mbox {${\rm Var}_{_{#1}}(#2)$}}
\def\ls[#1]{\mbox {$\xi^{^{#1}}$}}
\def\hom[#1,#2]{\mbox {${\rm Hom}({#1},{#2})$}}
\def\onvhom[#1,#2]{\mbox {${\rm Hom^{v}}(#1,#2)$}}
\def\onehom[#1,#2]{\mbox {${\rm Hom^{e}}(#1,#2)$}}
\def\core[#1]{\mbox {$#1^{^{\bullet}}$}}
\def\cay[#1,#2]{\mbox {${\rm Cay}({#1},{#2})$}}
\def\sch[#1,#2,#3]{\mbox {${\rm Sch}({#1},{#2},{#3})$}}
\def\cays[#1,#2]{\mbox {${\rm Cay_{s}}({#1},{#2})$}}
\def\dirc[#1]{\mbox {$\stackrel{\rightarrow}{C}_{_{#1}}$}}
\def\cycl[#1]{\mbox {${\bf Z}_{_{#1}}$}}

\begin{document}
\publicationdetails{20}{2018}{1}{9}{1368}
\maketitle

\begin{abstract}
Grebinski and Kucherov (1998)
and Alon et al. (2004-2005)
studied the problem of learning a hidden graph for some especial cases, such as hamiltonian cycle, cliques, stars, and matchings, which was motivated by some problems in 
chemical reactions, molecular biology and genome sequencing.
The present study aims to present a generalization of this problem. Graphs $G$ and $H$ were considered, by assuming  that $G$ includes exactly one defective subgraph isomorphic to $H$. 
The purpose is to find the defective subgraph by performing the minimum non-adaptive tests, 
where each test is an induced subgraph of the graph $G$ and the test is positive in the case of involving  at least one edge of the defective subgraph $H$.
We present the first upper bound for the number of non-adaptive tests to find the defective subgraph by using the 
symmetric and high probability variation of Lov\'asz Local Lemma.
Finally, we present the first non-adaptive randomized algorithm that finds the defective subgraph by at most ${3\over 2}$ times of this upper bound with high probability. 
\end{abstract}

\section{Introduction}
\label{sec:in}

In the classic \textit{group testing} problem which was first introduced by Dorfman {\rm \cite{dorfman1943}}, 
there is a set of $n$ items including at most $d$ defective items.
The purpose of this problem is to find the defective items with the minimum number of tests. Every test consists of 
some items and each test is positive if it includes at least one defective item. Otherwise, the test is negative.
There are two types of algorithms for the group testing problem, \textit{adaptive} and \textit{non-adaptive}. 
In adaptive algorithm, the outcome of previous tests can be used in the future tests and 
in non-adaptive algorithm all tests perform simultaneously and 
the defective items are obtained by considering results of all tests.

Regarding some extensions of classical group testing, we can refer to 
\textit{group testing on graphs}, \textit{complex group testing}, \textit{additive model}, \textit{inhibitor model}, etc.
(see 
{\rm \cite{du2000combinatorial-group-testing, du2006pooling-design-group-testing, survey-group-testing-2008}}
 for more information).
Aigner {\rm \cite{Aigner1986-search-problem-on-graphs}} 
proposed the problem of group testing on graphs, in which we look for 
one defective edge of the given graph $G$ by performing the minimum adaptive tests, 
where each test is an induced subgraph of the graph $G$ and the test is positive in the case of involving  the defective edge.

In the present paper, the problem of \textit{non-adaptive group testing on graphs} was considered by assuming 
that there is one defective subgraph (not~necessarily induced subgraph) of $G$ isomorphic to a graph $H$ 
and our purpose is to find the defective subgraph with minimum number of non-adaptive tests. 
Each test $F$ is an induced subgraph of $G$ and the test result is positive 
if and only if $F$ includes at least one edge of the defective subgraph. 
In this study we provide the first non-adaptive algorithm for this problem.
Our problem is a generalization of the problem of \textit{non-adaptive learning a hidden subgraph} studied in 
{\rm \cite{alon-learning-subgraph, alon-learning-matching, Grebinski1998-learning-hamiltonian-cycle}}.
In the problem of learning hidden graph, the graph $G$ is a complete graph.
In other words, let ${\cal H}$ be a family of labeled graphs on the set $V = \{1, 2, . . . , n\}$. 
In this problem the goal is to reconstruct a hidden graph $H \in {\cal H}$ by minimum number of tests, where a test  
${\cal F} \subset V$ is positive if the subgraph of $H$ induced by ${\cal F}$, 
contains at least one edge. Otherwise the test is negative.
Alon and Asodi {\rm \cite{alon-learning-subgraph}} 
for the problem of non-adaptive learning a general graph 
presented a lower bound based on the size of the maximum independent set. 
Their bound is almost tight for the random graph $G(n,{1\over 2})$.
Chang et al. {\rm \cite{learning-hidden-graph-Chang2014}} 
provided the best adaptive algorithm that learns the general hidden graph of $n$ vertices, 
with at most $m \log n + 10m + 3n$ tests when the hidden graph has
$m$ edges.
Also they proved 
$\displaystyle \Big\lceil \log \sum_{i=0}^m {{n\choose 2}\choose i} \Big\rceil$
adaptive tests are required to identify the hidden graph $H$(with $m$ edges) drawn
from the family of all graphs with $n$ vertices. 

The problem of learning a hidden graph was emphasized  in some models as follows:

\textit{K-vertex model}: In this model, each test has at most k vertices.

\textit{Additive model}: Based on  this model, the result of each test $F$ is the number of edges of $H$ induced by $F$. 
This model is mainly utilized in bioinformatics and was studied in 
{\rm \cite{2005-Bouvel-grebinski-survey, 2000-Grebinski-kucherov-additive}}.

\textit{Shortest path test}: 
In this model, each test ${u,v}$ indicates  the length of the shortest path 
between $u$ and $v$ in the hidden graph and if no path exists, it returns $\infty$. 
More information about this model and the result is given in {\rm \cite{2007-learning-edge-counting-srivastava}}.
Further, this model is regarded as a canonical model in the evolutionary tree literature 
{\rm \cite{1989-Hein-hidden-tree-additive-distance, 2003-king-learning-tree-distance, 
2007-Reyzin-learning-tree-distance-longest-path}}.

There are various families of hidden graphs to study.
However, a large number of recent studies have focused on 
hamiltonian cycles and matchings
{\rm \cite{alon-learning-matching, 2001-Beigel-genome-shotgun-sequencing, Grebinski1998-learning-hamiltonian-cycle}}, 
stars and cliques {\rm \cite{alon-learning-subgraph}}, 
graph of bounded degree {\rm \cite{2005-Bouvel-grebinski-survey,2000-Grebinski-kucherov-additive}}, 
general graphs {\rm \cite{2008-adaptive-learning-hidden-graph-per-edge, 2005-Bouvel-grebinski-survey} }.
Here, we present a short survey of known results on these problems by using adaptive and non-adaptive algorithms.


Grebinski and Kucherov {\rm \cite{Grebinski1998-learning-hamiltonian-cycle}} 
suggested an adaptive algorithm to learn a hidden Hamiltonian cycle by $2n \log n $
tests, which achieves the information lower bound for the number of tests needed. 
Further, Chang et al.{\rm \cite{2011-learning-hidden-graph-threshold}} improved their results to 
$(1+o(1))n\log n$.


Alon et al. {\rm \cite{alon-learning-matching}} proposed an upper bound $({1\over 2}+o(1)){n \choose 2}$ 
on learning a hidden matching using non-adaptive tests. 
Bouvel et al. {\rm \cite{2005-Bouvel-grebinski-survey}} developed an adaptive algorithm to learn a hidden 
matching with at most $(1+o(1))n\log n$ tests. 
In addition, Chang et al. {\rm \cite{2011-learning-hidden-graph-threshold}}
improved their result to $(1+o(1)){n\log n\over 2}$.


Alon and Asodi {\rm \cite{alon-learning-subgraph}} developed an upper bound $O(n\log^2 n)$ on learning a hidden clique 
using non-adaptive tests. 
Also they proved an upper bound $k^3\log n$ on 
learning a hidden $K_{1,k}$ using non-adaptive tests.
Bouvel et al. {\rm \cite{2005-Bouvel-grebinski-survey}} presented two adaptive algorithms to learn hidden 
star and hidden clique with at most $2n$ tests.
Chang et al. {\rm \cite{2011-learning-hidden-graph-threshold}} 
improved their results on learning hidden star and hidden clique 
to $(1+o(1))n$ and $n+\log n$, respectively.


Grebinski and Kucherov {\rm \cite{2000-Grebinski-kucherov-additive}} gave tight bound of
$\theta(dn)$ and $\theta({n^2\over \log n})$ non-adaptive tests on learning a hidden d-degree-bounded and general graphs  in additive model, respectively.
Angluin and Chen {\rm \cite{2008-adaptive-learning-hidden-graph-per-edge}} proved that a hidden general
graph can be identified with $12m \log n$ adaptive tests where $m$
(unknown) is the number of edges in the hidden graph. 


Group testing can be implemented in finding pattern in data, DNA library screening, and so on (see 
{\rm \cite{du2000combinatorial-group-testing, du2006pooling-design-group-testing, Macula2004-finding-patterns-in-data, 2000-survey-group-testing}} 
for an overview of results and more applications). 
Learning hidden graph, especially hamiltonian cycle and matchings, is mostly applied in genome sequencing, DNA physical mapping, chemical reactions and molecular biology (see
{\rm\cite{2008-adaptive-learning-hidden-graph-per-edge, 2011-learning-hidden-graph-threshold, Grebinski1998-learning-hamiltonian-cycle, Sorokin1996-application-physical-mapping}}
for more information about these applications).
Regarding the present study, the main motive behind investigating the problem of non-adaptive 
group testing on graphs is the application of this problem in chemical reactions.
In chemical reactions, we are dealing with a set of chemicals, some pairs of which may involve a reaction.
Moreover, before testing, we know some pairs have no reaction. 
When some chemicals are combined in one test, a reaction takes place if and only if at least one pair of the chemicals reacts in the test. The present study aimed to identify which pairs are reacted using as few tests as possible. 
Therefore, we can reformulate this problem as follows. Suppose that there are $n$ vertices and two vertices 
$u$ and $v$ are adjacent if and only if two chemicals $u$ and $v$ may involve a reaction.
The reaction of each pair of the chemicals indicates a defective edge and finding all there types of pairs is 
equal to find the defective subgraph. 
As we know some pairs have no reaction, the graph $G$ is not necessarily a complete graph.

%
%
%
\section{Notation}
Throughout this paper, we suppose that $H$ is a subgraph of $G$ with $k$ edges.
Moreover, we assume that $G$ contains exactly one defective subgraph 
isomorphic to $H$.

We denote the maximum degree of $H$ by $\Delta=\Delta(H)$. 
Also, $G[X]$ denotes the subgraph of $G$ induced by $X \cap V(G)$ and 
for any vertex $v \in G$, $N_H(v)$ stands for the set of neighbours of the vertex $v$ in the graph $H$.  
Hereafter, we assume that the subgraph $H$ has no isolated vertex, 
because in the problem of group testing on graphs, just edges are defective.

\section{Main result}
Throughout this paper, let $H_1,H_2,\ldots,H_m $ be all the subgraphs of $G$ isomorphic to $H$. 
For $1\leq l \leq t$, let ${\cal F}_l$ be a random set obtained by choosing each vertex of $V(G)$ randomly and independently with probability $p$.
For simplicity of notation we write $F_l$ as an induced subgraph of $G$ on the vertices of ${\cal F}_l$.
For the random subset ${\cal F}$ of $V(G)$, we say the random test $F=G[{\cal F}]$, 
distinguishes between two distinct subgraphs $H_i$ and $H_j$ if and only if 
$F$ contains an edge of one of the subgraphs $H_i$ and $H_j$ and contains no edge of the other.
In other words exactly one of the following events happens
$$E(F\cap H_i)\neq \varnothing \ and  \ E(F\cap H_j) = \varnothing$$ 
or 
$$E(F\cap H_i) = \varnothing  \ and \  E(F\cap H_j)\neq \varnothing $$
For any $i,j,l$, where $1\leq i \neq j \leq m$ and $1\leq l \leq t$, 
we define $A_{i,j}^l$ to be the event that the test $F_l$ cannot distinguish between $H_i$ and $H_j$. 
Also, let $A_{i,j}$ denote the event where there is 
no test $F\in \{F_1,F_2,\ldots,F_t\}$ that distinguishes between $H_i$ and $H_j$.
So we would like to bound the probability that none of the bad events $A_{i,j}$ occur.
%
%
%
In such cases, when there is some relatively small amount of dependence
between events, one can use a powerful generalization of the union bound,
known as the Lov\'asz Local Lemma.
The main device in establishing the Lov\'asz Local Lemma is a graph called the dependency graph.
Let $A_1,A_2,\ldots,A_n$ be events in an arbitrary probability space.
A graph $D=(V ,E)$  on the set of vertices $V=\{1,2,\ldots ,n \}$ is a dependency graph for events 
$A_1,A_2,\ldots,A_n$ if for each $1\leq i \leq n$ the event $A_i$ is mutually independent of all the events 
$\{A_j :\{i,j\}  \notin E \}$.
%
We state the Lov\'asz Local Lemma as follows.
\begin{alphlem}{\rm \cite{2008-probabilistic-method} }
{\rm (}Lov\'asz Local Lemma, Symmetric Case{\rm )}. 
Suppose that $A_1,A_2,\ldots ,A_n$ are events in a probability space with $Pr(A_i)\leq p$ for all i. 
If the maximum degree in the dependency graph of these events is $d$,
and if 
$ep(d+1)\leq 1$, then 
$$Pr \Big( \displaystyle \bigcap_{i=1}^n \overline{A_i} \Big) > 0,$$
where $e$ is the basis of the natural logarithm.
\end{alphlem}
%

To find the maximum degree in the dependency graph of the events $A_{i,j}$, we define the parameter  $r_G(H)$ as follows.
%
%
Set $r_G(H, H_i)$ to be the number of subgraphs of $G$ isomorphic to $H$ have common vertex with $H_i$, i.e., 
$r_G(H,H_i)=|\{H_j : 1\leq j \leq m, j\neq i, V(H_i)\cap V(H_j)\neq \varnothing \}|$.
Also, define $$r_G(H)=\displaystyle \max_{1\leq i\leq m}  r_G(H, H_i).$$

In Theorem~\ref{main-theorem}, we show that 
there are $t$ tests, $F_1, F_2,\ldots, F_t$, such that for every $i$ and $j$,
there is a test $F\in \{F_1,F_2,\ldots,F_t\}$ that distinguishes between $H_i$ and $H_j$.  
%
%
So if $H_i$ is the defective subgraph, then for every non-defective subgraph $H_j$, there exists a 
test $F\in\{F_1,F_2,\ldots,F_t\}$ that distinguishes between the defective subgraph $H_i$ and non-defective subgraph $H_j$. 
Therefore, by these tests we can find the defective subgraph.


\begin{thm}\label{main-theorem}
Let $H$ be the defective subgraph of $G$ 
and $H_1,H_2,\ldots,H_m$ be all the subgraphs of $G$ isomorphic to $H$. 
There are $t$ induced subgraphs $F_1,\ldots,F_t$ of $G$ such that for each pair of $H_i$ and $H_j$, 
at least one of $F_1,\ldots,F_t $ can distinguish between $H_i$ and $H_j$, where 
$k=|E(H)|$, $\Delta=\Delta (H)$, 
$$t=1 + \left\lceil\frac{\ln (4 e r_G(H) ) + \ln m }{ \ln {1 \over 1 - P_{k,\Delta}} }\right\rceil ,$$
 $P_{k,\Delta}= \frac{1}{2k \Delta } 
\left(1-\frac{1}{ 2 \Delta } \right)^{2 \Delta - 1} 
\left( 1-\sqrt{{1 \over 2k \Delta }}  \left(1-\frac{1}{ 2 \Delta } \right)^{ \Delta - 1}  \right)^{2\Delta -2}$, and
$e$ is the basis of the natural logarithm.
\end{thm}
%

In order to prove Theorem~\ref{main-theorem}, 
first we should find the probability that tests 
$F_1,F_2,\ldots,F_t$, distinguish between each pair of subgraphs $H_i$ and $H_j$.
Thus, finding the upper bound for the probability of occurring the bad event $A_{i,j}$ is essential. 
Accordingly, we should find the lower bound of probability that the random test $F_l$ can distinguish between two subgraphs $H_i$ and $H_j$. 

In the next theorem, based on some following lemmas, we show that the probability of distinguishing between $H_i$ and $H_j$ has the minimum value whenever  
$V(H_i)=V(H_j)$ and $|E(H_i)\setminus E(H_j) |=1$. 

\begin{thm}\label{lower bound-A-{i,j}^l}
Let $k=|E(H)|$ and $\Delta=\Delta (H)$.
For every $1\leq i \neq j\leq m $ and $1 \leq l \leq t$, we have 
\begin{equation}\label{probability-distinguish-dense}
Pr \left( \overline{A_{i,j}^l} \right) \geq 2p^2 (1-p)^{ 2 \Delta } (1- \epsilon),
\end{equation} 
where $p=\sqrt{ { \epsilon \over k } } \left( 1- \epsilon \right)^{ \Delta -1}.$
\end{thm}


\begin{lem}\label{lovasz-independent-set}
Let $T$ be a graph with $n$ vertices, $k$ edges, and maximum degree $\Delta$.
Pick, randomly and independently, each vertex of $T$ with probability $p$, where
$p=\sqrt{ \frac{\epsilon}{k} } (1-\frac{\epsilon}{k})^{(\Delta -1)}$. 
If $F$ is the set of all chosen vertices, 
then $T[F]$ has no edges,
with probability at least $1-\epsilon$.
\end{lem}

To prove this lemma, we need high probability variation of Lov\'asz Local Lemma.

\begin{alphlem}\label{high-probability-lovasz}{\rm \cite{combinatorial-array-stinson}}
Let $B_1,B_2,\ldots ,B_k$ be events in a probability space. Suppose that each event $B_i$ is 
independent of all the events $B_j$ but at most d. 
For $1 \leq i \leq k$ and $0 < \epsilon <1$, if 
$\displaystyle Pr(B_i) \leq {\epsilon \over k }(1- {\epsilon \over k} )^d$, then
$ \displaystyle Pr \Big( \bigcap_{i=1}^k \overline{B_i} \Big) > 1- \epsilon $.
\end{alphlem}
\begin{proof}[of Lemma~\ref{lovasz-independent-set}]
Let $E(T) =\{ e_1,e_2,\ldots,e_k \}$.
For $1\leq i \leq k$, we define $B_i$ to be the event that $e_i \in E(T[F])$, so $Pr(B_i)=p^2$.
Since vertices are chosen randomly and independently, the event $B_i$ is independent of the event $B_j$ if and only if
edges $e_i$ and $e_j$ have no common vertex.
So the maximum degree of the dependency graph is at most $2 ( \Delta -1 ) $.
Since $\displaystyle p^2 \leq {\epsilon \over k }\left( 1- {\epsilon \over k} \right)^ {2 (\Delta -1)}$,
by Lemma~\ref{high-probability-lovasz}, 
$\displaystyle Pr\Big( \bigcap_{i=1}^k \overline{B_i} \Big) > 1- \epsilon$. 
Hence, $T[F]$ has no edges,
with probability at least $1-\epsilon$.
\end{proof}

To find the probability of distinguishing between $H_i$ and $H_j$ and 
then prove Theorem~\ref{lower bound-A-{i,j}^l}, we consider the following three cases.
\begin{enumerate}[{Case }1:]
\item $V(H_i)=V(H_j)$, $|E(H_i)\setminus E(H_j) |=1$.
\item $|V(H_i)\setminus V(H_j)| \geq 1$.
\item The induced subgraph on $V(H_i) - V(H_j)$ has at least one edge.
\end{enumerate}
%
\begin{lem}\label{lem-similar1}
If $V(H_i)=V(H_j)$ and $|E(H_i)\setminus E(H_j) |=1$, 
then
$$
Pr\Big( E(F_l \cap H_i)\neq \varnothing , E(F_l \cap H_j) = \varnothing \Big) \geq p^2 (1-p)^{2 \Delta  } (1- \epsilon),
$$ 
where
$p=\sqrt{ { \epsilon \over k } }\left( 1-  \epsilon \right)^{ \Delta - 1}$.
\end{lem}

\begin{proof}
{
Let $e=\{ u,v \}\in E(H_i)\setminus E(H_j)$. 
Consider the induced subgraph $H'$ of $G$, where 
$V(H')=V(H_j) \setminus \Big( u \cup v \cup N_{H_j}(u) \cup N_{H_j}(v) \Big)$.
Note that if $u,v \in {\cal F}_l $ and $H_j\cap F_l$ has no edges of $H_j$, then 
$E(F_l \cap H_i)\neq \varnothing $ and $E(F_l \cap H_j) = \varnothing $.
Also, one can see that 
$u,v \in {\cal F}_l $ and $H_j[F_l]$ 
has no edges 
if the following events hold
\begin{enumerate}
\item $u,v \in {\cal F}_l $, 
\item $N_{H_j}(u) \cap {\cal F}_l = \varnothing$ and $ N_{H_j}(v) \cap {\cal F}_l= \varnothing $,
\item $H'[{\cal F}_l]$ has no edges.
\end{enumerate}
It is straightforward to check that the aforementioned events are independent. Also, 
one can see that the event  $u,v \in {\cal F}_l $ occurs with probability $p^2$. 
Since $ | N_{H_j}(u) \cup N_{H_j}(v) | \leq 2 \Delta $, we have 
$$
\begin{array}{rcl}
Pr \Big(N_{H_j}(u) \cap {\cal F}_l = \varnothing ,  N_{H_j}(v) \cap {\cal F}_l = \varnothing \Big) & = & \\
 Pr \Big( {\cal F}_l \cap \left( N_{H_j}(u) \cup N_{H_j}(v) \setminus \{u, v\} \right)   =\varnothing \Big)
& \geq &  (1-p)^{2\Delta}.
\end{array}$$

Set $E(H')=k'$.
If $k' = 0 $, then $F_l \cap H'$ has no edges and $Pr \big( E( F_l \cap H' )= \varnothing \big) = 1 $.
Suppose that $k' \geq 1$. Since $k\geq k'$, we have 
$\displaystyle p^2={\epsilon \over k}(1-\epsilon)^{2 \Delta -2} \leq {\epsilon \over k'}(1-{\epsilon \over k'})^{2\Delta - 2 }$.
Each vertex of the induced subgraph $H'$ is chosen with probability $p$.
So by 
 Lemma~\ref{lovasz-independent-set}, 
the induced subgraph on ${\cal F}_l \cap V(H')$ has no edges,
with probability at least $1- \epsilon$. In other words,
$Pr \big( E(F_l \cap H')= \varnothing \big) \geq 1- \epsilon $. 
Since the events are independent, we have 
$$Pr \Big(E(F_l \cap H_i)\neq \varnothing , E(F_l \cap H_j) = \varnothing \Big) \geq
p^2 (1-p)^{2 \Delta } (1- \epsilon),
$$
as desired.
}
\end{proof}

\begin{lem}\label{lem-similar2}
If $|V(H_i)\setminus V(H_j)| \geq 1$, then 
$$
Pr\Big( E(F_l \cap H_i)\neq \varnothing , E(F_l \cap H_j) = \varnothing \Big) \geq p^2 (1-p)^{\Delta } (1- \epsilon),
$$
where
$p=\sqrt{ { \epsilon \over k } }\left( 1-  \epsilon \right)^{ \Delta -1}$.
\end{lem}

\begin{proof}
{
Since $H$ has no isolated vertex, 
there exists at least one edge $e=\{ u,v\} \in E(H_i)\setminus E(H_j)$. 
Let $v \in V(H_i) \cap V(H_j) $ and $u \in V(H_i) \setminus V(H_j)$. 
Suppose that $H'$ is an induced subgraph of $H_j$, where 
$V(H')=V(H_j) \setminus \left( v \cup N(v) \right)$. Set $|E(H')|=k'$.
%
Similar to the proof of Lemma~\ref{lem-similar1}, 
$E(F_l \cap H_i)\neq \varnothing $ and $E(F_l \cap H_j) = \varnothing $ if 
the following independent events hold
\begin{enumerate}
\item $u,v \in {\cal F}_l $, 
\item $ N_{H_j}(v) \cap {\cal F}_l= \varnothing $,
\item $H'[{\cal F}_l]$ has no edges.
\end{enumerate}


Since $ |N_{H_j}(v)| \leq \Delta$, the probability that 
$N_{H_j}(v) \cap {\cal F}_l =\varnothing$
is at least
$(1-p)^{\Delta } $. 
The rest of proof is similar to Lemma~\ref{lem-similar1}, so 
$$
Pr \Big(E(F_l \cap H_i)\neq \varnothing , E(F_l \cap H_j) = \varnothing \Big) \geq
p^2 (1-p)^{ \Delta } (1- \epsilon),
$$



as desired.
}
\end{proof}
%
\begin{lem}\label{lem-similar3}
If the induced subgraph on $V(H_i)\setminus V(H_j)$ has at least one edge, then
$$
Pr\Big( E(F_l \cap H_i)\neq \varnothing , E(F_l \cap H_j) = \varnothing \Big) \geq p^2 (1- \epsilon), 
$$
where $p=\sqrt{ { \epsilon \over k } }\left( 1- { \epsilon } \right)^{ \Delta - 1}$.

\end{lem}

\begin{proof}
{
Let $e=(u,v) \in E(H_i)\setminus E(H_j)$.
If the following independent events hold
\begin{enumerate}
\item $u,v \in {\cal F}_l$,
\item $H_j[{\cal F}_l]$ has no edges,
\end{enumerate}
then 
$E(F_l \cap H_i)\neq \varnothing $ and $E(F_l \cap H_j) = \varnothing $. 
%
Since $\displaystyle p^2 = { \epsilon \over k } \left( 1- { \epsilon } \right)^{2 \Delta -2}
 \leq { \epsilon \over k } \left( 1- { \epsilon \over k } \right)^{2 \Delta -2}$,
 by Lemma~\ref{lovasz-independent-set}, 
$Pr\left( E(F_l \cap H_j) = \varnothing \right) \geq 1- \epsilon $.
Also one can see that 
$$Pr\Big(E(F_l \cap H_i)\neq \varnothing \Big) \geq Pr\Big(e \in E(F_l) \Big) = Pr\big(u,v \in {\cal F}_l \big)= p^2.$$
%
%
Consequently,
$
Pr \big(E(F_l \cap H_i)\neq \varnothing , E(F_l \cap H_j) = \varnothing \big) \geq  p^2 (1- \epsilon).
$}
\end{proof}



\begin{proof}[of Theorem~\ref{lower bound-A-{i,j}^l}]
Let $E(H_i)\cap E(H_j)= \{f_1,f_2,\ldots,f_r \}$ and $E(H_i) \setminus E(H_j)= \{e_1,e_2,\ldots,e_{k-r}\}$. 
As previously mentioned, the event $ \overline{ A_{i,j}^l }$ occurs if and only if 
the test $F_l$ distinguish between $H_i$ and $H_j$. In other words, 
$$E(F_l \cap H_i)\neq \varnothing \ and \ E(F_l \cap H_j)= \varnothing$$ or 
$$E(F_l \cap H_j)\neq \varnothing \ and \ E(F_l \cap H_i)= \varnothing.$$
It is easy to check that 
$$Pr\left(\overline{ A_{i,j}^l }\right)=
Pr\Big(E(F_l \cap H_i)\neq \varnothing , E(F_l \cap H_j)= \varnothing\Big)+
Pr\Big(E(F_l \cap H_j)\neq \varnothing , E(F_l \cap H_i)= \varnothing\Big).$$
In the following we prove 
$\displaystyle Pr\big(E(F_l \cap H_i)\neq \varnothing , E(F_l \cap H_j)= \varnothing\big)\geq p^2 (1-p)^{ 2 \Delta } (1- \epsilon)$ and with the completely similar proof we can prove 
$\displaystyle Pr\big(E(F_l \cap H_j)\neq \varnothing , E(F_l \cap H_i)= \varnothing\big)\geq p^2 (1-p)^{ 2 \Delta } (1- \epsilon)$.

 It is easy to check, for every $1\leq q \leq k - r $,
$$ 
Pr\Big( E(F_l \cap H_i)\neq \varnothing , E(F_l \cap H_j)=\varnothing \Big)
\geq Pr\Big(e_q\in E(F_l \cap H_i) , E(F_l \cap H_j)=\varnothing \Big). 
$$
So to find the lower bound for this probability, we need to consider the following three cases.
\begin{enumerate}[{Case }1:]
\item $V(H_i)=V(H_j)$, $|E(H_i)\setminus E(H_j) |=1$.

By Lemma~\ref{lem-similar1}, it is clear 
$$Pr\big( E(F_l \cap H_i)\neq \varnothing , E(F_l \cap H_j) = \varnothing \big) \geq p^2 (1-p)^{2\Delta  } (1- \epsilon).$$

\item $|V(H_i)\setminus V(H_j)| \geq  1$.

By Lemma~\ref{lem-similar2}, we have 
$$Pr\big( E(F_l \cap H_i)\neq \varnothing , E(F_l \cap H_j) = \varnothing \big) 
\geq p^2 (1-p)^{\Delta } (1- \epsilon)
 \geq p^2 (1-p)^{2 \Delta } (1- \epsilon).$$

\item The induced subgraph on $V(H_i) - V(H_j)$ has at least one edge.

By Lemma~\ref{lem-similar3}, 
$$ Pr\big( E(F_l \cap H_i)\neq \varnothing , E(F_l \cap H_j) = \varnothing \big)
\geq p^2 (1 - \epsilon ) \geq  p^2 (1-p)^{ 2 \Delta } (1- \epsilon).$$
\end{enumerate}
So for every $1 \leq i\neq j \leq m$ and $1 \leq l \leq t$, 
$Pr \left( \overline{A_{i,j}^l} \right) \geq 2p^2 (1-p)^{ 2 \Delta } (1- \epsilon)$.
\end{proof}

In order to prove Theorem~\ref{main-theorem}, we present an upper bound for 
the probability of occurring the bad events $A_{i,j}$ for every $1\leq i\neq j \leq m$.


\begin{thm}\label{theorem-upper-bound-A-{i,j}}
Let $k=|E(H)|$ and $\Delta=\Delta (H)$.
For every $1\leq i \neq j\leq m $, we have 
\begin{equation}\label{upper bound-A-{i,j}}
Pr(A_{i,j}) \leq (1-P_{k,\Delta})^t,
\end{equation} 
where 
$P_{k,\Delta}= \frac{1}{2k \Delta } 
\left(1-\frac{1}{ 2 \Delta } \right)^{2 \Delta - 1} 
\left( 1-\sqrt{{1 \over 2k \Delta }}  \left(1-\frac{1}{ 2 \Delta } \right)^{ \Delta - 1}  \right)^{2\Delta -2}$.
\end{thm}
\begin{proof}
{
Since ${\cal F}_1, {\cal F}_2, \ldots , {\cal F}_t \subset V(G)$ are chosen randomly and independently, 
the events $A_{i,j}^1, \ldots , A_{i,j}^t$ are mutually independent.
So
$$
\displaystyle Pr(A_{i,j})= Pr(A_{i,j}^l)^t.
$$
%
By Theorem~\ref{lower bound-A-{i,j}^l}, we have
$Pr \left( \overline{ A_{i,j}^l } \right) 
\geq 2 p^2 (1-p)^{ 2 \Delta } (1- \epsilon).$
According to $p=\sqrt{ { \epsilon \over k } }\left( 1- { \epsilon } \right)^{ \Delta - 1}$,
we set $\epsilon ={3\over \Delta}$  
to almostly maximize the lower bound of good events $\overline{A_{i,j}^l}$.
%
%
%
So
$Pr \left( \overline{A_{i,j}^l} \right) \geq P_{k,\Delta}$, where
$$ P_{k,\Delta}= \frac{6}{k \Delta } 
\left(1-\frac{3}{  \Delta } \right)^{2 \Delta - 1} 
\left( 1-\sqrt{{3 \over k \Delta }}  \left(1-\frac{3}{  \Delta } \right)^{ \Delta - 1}  \right)^{2\Delta }.$$
Therefore, $Pr(A_{i,j})= Pr(A_{i,j}^l)^t \leq (1-P_{k,\Delta})^t.$
}
\end{proof}

Now, we can prove Theorem \ref{main-theorem}.

\begin{proof}[of Theorem~\ref{main-theorem}]
By Theorem~\ref{theorem-upper-bound-A-{i,j}}, 
for every $1\leq i\neq j \leq m$, $Pr( A_{i,j} )\leq (1 - P_{k,\Delta})^t$. 
%
%
%
Now we prove that if $\displaystyle t > \frac{\ln (4 e r_G(H) ) + \ln m }{ \ln {1 \over 1 - P_{k,\Delta}} }$,
then by Lov\'asz Local Lemma, with positive probability no event $A_{i,j}$ occurs.
%
%
%

We construct the dependency graph whose vertices are the events $A_{i,j}$, where $1\leq i,j \leq m$. 
Two events $A_{i,j}$ and $A_{i',j'}$ are adjacent if and only if 
$ \Big( V(H_i) \cup V(H_j) \Big) \cap \Big( V(H_{i'}) \cup V(H_{j'}) \Big) \neq \varnothing $.
Remember that $r_G(H)=\displaystyle \max_i r_G(H, H_i)$, 
where $r_G(H, H_i)$ is the number of subgraphs of $G$ isomorphic to $H$ including common vertex with $H_i$.
%
%
%
For the fixed $A_{i,j}$, there are at most $r_G(H)$ subgraph $H_{i'}$ isomorphic to $H$ 
such that $V(H_i)\cap V(H_{i'}) \neq \varnothing$. We can choose $H_{j'}$ with $m-1$ ways. 
So it is easy to check that the maximum degree in the dependency graph is at most $4r_G(H)(m-1)$. 
%
Accordingly, if 
$$t >\frac{\ln (4 e r_G(H) ) + \ln m }{ \ln {1 \over 1 - P_{k,\Delta}} },$$ 
then
$e \left(1- P_{k,\Delta} \right)^t \big(4 r_G(H) (m-1) + 1 \big) < 1 $,
and by Lov\' asz Local Lemma 
$$
Pr \Big(\displaystyle \bigcap_{i,j} \overline{ A_{i,j} } \Big) >0.
$$
%
Therefore, if 
$\displaystyle t=1 + \lceil\frac{\ln (4 e r_G(H) ) + \ln m }{ \ln {1 \over 1 - P_{k,\Delta}} }\rceil $,
then with positive probability no event $A_{i,j}$ occurs.
Thus, for each pair of $H_i$ and $H_j$ 
there is a test $F\in \{F_1, F_2,\ldots, F_t\}$ that distinguishes between $H_i$ and $H_j$. 
%
%
%
%
%
%
%
%
\end{proof}

We can obtain $t=1+\lceil\frac{2 \ln m}{\ln {1 \over {1 - P_{k,\Delta}} }} \rceil$ 
if we use union bound. In fact the Lov\'asz Local Lemma is better when 
the dependencies between events are rare. 

Based on this theorem there are $t$ tests which distinguish between each pair of 
$H_i$ and $H_j$ with positive probability. 
However, an algorithm is essential to find these tests with high probability if we are interested in finding these tests.
%

\begin{thm}\label{Algorithm}
Let $H$ be the defective subgraph of $G$ with $k$ edges. 
If $\displaystyle t={\ln {m^2\over \delta}\over \ln {1\over 1 -P_{k,\Delta}}} $,
we can find this defective subgraph by $t$ tests with probability at least $1-\delta$, where 
$\Delta=\Delta (H)$ and 
$$P_{k,\Delta}= \frac{1}{2k \Delta } 
\left(1-\frac{1}{ 2 \Delta } \right)^{2 \Delta - 1} 
\left( 1-\sqrt{{1 \over 2k \Delta }}  \left(1-\frac{1}{ 2 \Delta } \right)^{ \Delta - 1}  \right)^{2\Delta -2}.$$
\end{thm}
\begin{proof}
{
By Theorem~\ref{theorem-upper-bound-A-{i,j}} and the union bound we know 
$$\displaystyle Pr\big(\bigcup_{1\leq i<j\leq m} A_{i,j} \big)\leq m^2 (1-P_{k,\Delta})^t.$$
Thus, this upper bound becomes close to zero if $t$ is large enough. 
It is easy to check if $\displaystyle t = {\ln {m^2\over \delta}\over \ln {1\over 1 -P_{k,\Delta}}}$, then 
$m^2 (1-P_{k,\Delta})^t = \delta$.
In other words, we can distinguish between each pair of $H_i$ and $H_j$ with probability at least $1-\delta$ 
if we choose tests randomly and independently.
}
\end{proof}

If we set $\delta={1\over m}$, then for sufficiently large $m$, 
we can find the defective subgraph with 
$\frac{3 \ln m}{\ln {1 \over {1 - P_{k,\Delta}} }}$ tests with high probability.



\section{Concluding remarks}
In the present paper we assume that the graph $G$ includes few edges since the 
Lov\'asz Local Lemma is more powerful when the dependencies between events are rare. 
In the graph $G$ with $O(n^2)$ edges, the parameter $r_G(H)$ is high, which means 
it is better to use the union bound. 
In this case, there are $\displaystyle 1+\lceil\frac{2 \ln m}{\ln {1 \over {1 - P_{k,\Delta}} }} \rceil$ tests that find the defective subgraph non-adaptively.

Finally, if we consider dense and sparse defective subgraph separately, 
we can obtain a  better upper bound for the number of tests in the case of sparse defective subgraph.
%
\acknowledgements
\label{sec:ack}
The present study is a part of Hamid Kameli's Ph.D. Thesis.
The author would like to express his deepest gratitude to Professor Hossein Hajiabolhassan 
for introducing a generalization of learning a hidden subgraph problem and for his
precious comments and discussion.

\bibliographystyle{abbrvnat}

\begin{thebibliography}{23}
\providecommand{\natexlab}[1]{#1}
\providecommand{\url}[1]{\texttt{#1}}
\expandafter\ifx\csname urlstyle\endcsname\relax
  \providecommand{\doi}[1]{doi: #1}\else
  \providecommand{\doi}{doi: \begingroup \urlstyle{rm}\Url}\fi

\bibitem[Aigner(1986)]{Aigner1986-search-problem-on-graphs}
M.~Aigner.
\newblock Search problems on graphs.
\newblock \emph{Discrete Applied Mathematics}, 14\penalty0 (3):\penalty0 215 --
  230, 1986.
\newblock ISSN 0166-218X.
\newblock \doi{http://dx.doi.org/10.1016/0166-218X(86)90026-0}.
\newblock URL
  \url{http://www.sciencedirect.com/science/article/pii/0166218X86900260}.

\bibitem[Alon and Asodi(2005)]{alon-learning-subgraph}
N.~Alon and V.~Asodi.
\newblock Learning a hidden subgraph.
\newblock \emph{SIAM Journal on Discrete Mathematics}, 18\penalty0
  (4):\penalty0 697--712, 2005.

\bibitem[Alon and Spencer(2008)]{2008-probabilistic-method}
N.~Alon and J.~H. Spencer.
\newblock \emph{The probabilistic method}.
\newblock Wiley-Interscience Series in Discrete Mathematics and Optimization.
  John Wiley \& Sons Inc., Hoboken, NJ, third edition, 2008.
\newblock ISBN 978-0-470-17020-5.

\bibitem[Alon et~al.(2004)Alon, Beigel, Kasif, Rudich, and
  Sudakov]{alon-learning-matching}
N.~Alon, R.~Beigel, S.~Kasif, S.~Rudich, and B.~Sudakov.
\newblock Learning a hidden matching.
\newblock \emph{SIAM Journal on Computing}, 33\penalty0 (2):\penalty0 487--501,
  2004.

\bibitem[Angluin and Chen(2008)]{2008-adaptive-learning-hidden-graph-per-edge}
D.~Angluin and J.~Chen.
\newblock Learning a hidden graph using queries per edge.
\newblock \emph{Journal of Computer and System Sciences}, 74\penalty0
  (4):\penalty0 546 -- 556, 2008.
\newblock ISSN 0022-0000.
\newblock \doi{http://dx.doi.org/10.1016/j.jcss.2007.06.006}.
\newblock URL
  \url{http://www.sciencedirect.com/science/article/pii/S0022000007000724}.

\bibitem[Beigel et~al.(2001)Beigel, Alon, Kasif, Apaydin, and
  Fortnow]{2001-Beigel-genome-shotgun-sequencing}
R.~Beigel, N.~Alon, S.~Kasif, M.~S. Apaydin, and L.~Fortnow.
\newblock An optimal procedure for gap closing in whole genome shotgun
  sequencing.
\newblock In \emph{Proceedings of the Fifth Annual International Conference on
  Computational Biology}, RECOMB '01, pages 22--30, New York, NY, USA, 2001.
  ACM.
\newblock ISBN 1-58113-353-7.
\newblock \doi{10.1145/369133.369152}.
\newblock URL \url{http://doi.acm.org/10.1145/369133.369152}.

\bibitem[Bouvel et~al.(2005)Bouvel, Grebinski, and
  Kucherov]{2005-Bouvel-grebinski-survey}
M.~Bouvel, V.~Grebinski, and G.~Kucherov.
\newblock \emph{Combinatorial Search on Graphs Motivated by Bioinformatics
  Applications: A Brief Survey}, pages 16--27.
\newblock Springer Berlin Heidelberg, Berlin, Heidelberg, 2005.
\newblock ISBN 978-3-540-31468-4.
\newblock \doi{10.1007/11604686_2}.
\newblock URL \url{http://dx.doi.org/10.1007/11604686_2}.

\bibitem[Chang et~al.(2011)Chang, Chen, Fu, and
  Shi]{2011-learning-hidden-graph-threshold}
H.~Chang, H.-B. Chen, H.-L. Fu, and C.-H. Shi.
\newblock Reconstruction of hidden graphs and threshold group testing.
\newblock \emph{Journal of Combinatorial Optimization}, 22\penalty0
  (2):\penalty0 270--281, 2011.
\newblock ISSN 1382-6905.
\newblock \doi{10.1007/s10878-010-9291-0}.
\newblock URL \url{http://dx.doi.org/10.1007/s10878-010-9291-0}.

\bibitem[Chang et~al.(2014)Chang, Fu, and
  Shih]{learning-hidden-graph-Chang2014}
H.~Chang, H.-L. Fu, and C.-H. Shih.
\newblock Learning a hidden graph.
\newblock \emph{Optimization Letters}, 8\penalty0 (8):\penalty0 2341--2348, Dec
  2014.
\newblock ISSN 1862-4480.
\newblock \doi{10.1007/s11590-014-0751-9}.
\newblock URL \url{https://doi.org/10.1007/s11590-014-0751-9}.

\bibitem[Deng et~al.(2004)Deng, Stinson, and Wei]{combinatorial-array-stinson}
D.~Deng, D.~Stinson, and R.~Wei.
\newblock The {L}ov\'asz local lemma and its applications to some combinatorial
  arrays.
\newblock \emph{Designs, Codes and Cryptography}, 32\penalty0 (1-3):\penalty0
  121--134, 2004.
\newblock ISSN 0925-1022.
\newblock \doi{10.1023/B:DESI.0000029217.97956.26}.
\newblock URL \url{http://dx.doi.org/10.1023/B%3ADESI.0000029217.97956.26}.

\bibitem[Dorfman(1943)]{dorfman1943}
R.~Dorfman.
\newblock The detection of defective members of large populations.
\newblock \emph{Ann. Math. Statist.}, 14\penalty0 (4):\penalty0 436--440, 12
  1943.
\newblock \doi{10.1214/aoms/1177731363}.
\newblock URL \url{http://dx.doi.org/10.1214/aoms/1177731363}.

\bibitem[Du and Hwang(2000)]{du2000combinatorial-group-testing}
D.~Du and F.~Hwang.
\newblock \emph{Combinatorial Group Testing and Its Applications}.
\newblock Applied Mathematics. World Scientific, 2000.
\newblock ISBN 9789810241070.
\newblock URL \url{http://books.google.com/books?id=KW5-CyUUOggC}.

\bibitem[Du and Hwang(2006)]{du2006pooling-design-group-testing}
D.~Du and F.~Hwang.
\newblock \emph{Pooling Designs and Nonadaptive Group Testing: Important Tools
  for DNA Sequencing}.
\newblock Series on applied mathematics. World Scientific, 2006.
\newblock ISBN 9789812568229.
\newblock URL \url{http://books.google.com/books?id=AJcoAQAAMAAJ}.

\bibitem[Grebinski and
  Kucherov(1998)]{Grebinski1998-learning-hamiltonian-cycle}
V.~Grebinski and G.~Kucherov.
\newblock Reconstructing a hamiltonian cycle by querying the graph: Application
  to {DNA} physical mapping.
\newblock \emph{Discrete Applied Mathematics}, 88\penalty0 (1–3):\penalty0
  147 -- 165, 1998.
\newblock ISSN 0166-218X.
\newblock \doi{http://dx.doi.org/10.1016/S0166-218X(98)00070-5}.
\newblock URL
  \url{http://www.sciencedirect.com/science/article/pii/S0166218X98000705}.

\bibitem[Grebinski and Kucherov(2000)]{2000-Grebinski-kucherov-additive}
V.~Grebinski and G.~Kucherov.
\newblock Optimal reconstruction of graphs under the additive model.
\newblock \emph{Algorithmica}, 28\penalty0 (1):\penalty0 104--124, 2000.
\newblock ISSN 1432-0541.
\newblock \doi{10.1007/s004530010033}.
\newblock URL \url{http://dx.doi.org/10.1007/s004530010033}.

\bibitem[Hein(1989)]{1989-Hein-hidden-tree-additive-distance}
J.~J. Hein.
\newblock An optimal algorithm to reconstruct trees from additive distance
  data.
\newblock \emph{Bulletin of Mathematical Biology}, 51\penalty0 (5):\penalty0
  597--603, 1989.
\newblock ISSN 1522-9602.
\newblock \doi{10.1007/BF02459968}.
\newblock URL \url{http://dx.doi.org/10.1007/BF02459968}.

\bibitem[Hwang(2008)]{survey-group-testing-2008}
H.-B. C. F.~K. Hwang.
\newblock A survey on nonadaptive group testing algorithms through the angle of
  decoding.
\newblock \emph{Journal of Combinatorial Optimization}, 15, 01 2008.
\newblock \doi{10.1007/s10878-007-9083-3}.
\newblock URL
  \url{http://libgen.org/scimag/index.php?s=10.1007/s10878-007-9083-3}.

\bibitem[King et~al.(2003)King, Zhang, and
  Zhou]{2003-king-learning-tree-distance}
V.~King, L.~Zhang, and Y.~Zhou.
\newblock On the complexity of distance-based evolutionary tree reconstruction.
\newblock In \emph{Proceedings of the Fourteenth Annual ACM-SIAM Symposium on
  Discrete Algorithms}, SODA '03, pages 444--453, Philadelphia, PA, USA, 2003.
  Society for Industrial and Applied Mathematics.
\newblock ISBN 0-89871-538-5.
\newblock URL \url{http://dl.acm.org/citation.cfm?id=644108.644179}.

\bibitem[Macula and Popyack(2004)]{Macula2004-finding-patterns-in-data}
A.~J. Macula and L.~J. Popyack.
\newblock A group testing method for finding patterns in data.
\newblock \emph{Discrete Applied Mathematics}, 144\penalty0 (1–2):\penalty0
  149 -- 157, 2004.
\newblock ISSN 0166-218X.
\newblock \doi{http://dx.doi.org/10.1016/j.dam.2003.07.009}.
\newblock URL
  \url{http://www.sciencedirect.com/science/article/pii/S0166218X04001933}.

\bibitem[Ngo and Du(2000)]{2000-survey-group-testing}
H.~Q. Ngo and D.-Z. Du.
\newblock A survey on combinatorial group testing algorithms with applications
  to {DNA} library screening.
\newblock In \emph{Discrete mathematical problems with medical applications},
  volume~55 of \emph{DIMACS Ser. Discrete Math. Theoret. Comput. Sci.}, pages
  171--182. Amer. Math. Soc., 2000.

\bibitem[Reyzin and
  Srivastava(2007{\natexlab{a}})]{2007-learning-edge-counting-srivastava}
L.~Reyzin and N.~Srivastava.
\newblock \emph{Algorithmic Learning Theory: 18th International Conference, ALT
  2007, Sendai, Japan, October 1-4, 2007. Proceedings}, chapter Learning and
  Verifying Graphs Using Queries with a Focus on Edge Counting, pages 285--297.
\newblock Springer Berlin Heidelberg, Berlin, Heidelberg, 2007{\natexlab{a}}.
\newblock ISBN 978-3-540-75225-7.
\newblock \doi{10.1007/978-3-540-75225-7_24}.
\newblock URL \url{http://dx.doi.org/10.1007/978-3-540-75225-7_24}.

\bibitem[Reyzin and
  Srivastava(2007{\natexlab{b}})]{2007-Reyzin-learning-tree-distance-longest-p%
ath}
L.~Reyzin and N.~Srivastava.
\newblock On the longest path algorithm for reconstructing trees from distance
  matrices.
\newblock \emph{Inf. Process. Lett.}, 101\penalty0 (3):\penalty0 98--100, Feb.
  2007{\natexlab{b}}.
\newblock ISSN 0020-0190.
\newblock \doi{10.1016/j.ipl.2006.08.013}.
\newblock URL \url{http://dx.doi.org/10.1016/j.ipl.2006.08.013}.

\bibitem[Sorokin et~al.(1996)Sorokin, Lapidus, Capuano, Galleron, Pujic, and
  Ehrlich]{Sorokin1996-application-physical-mapping}
A.~Sorokin, A.~Lapidus, V.~Capuano, N.~Galleron, P.~Pujic, and S.~D. Ehrlich.
\newblock A new approach using multiplex long accurate pcr and yeast artificial
  chromosomes for bacterial chromosome mapping and sequencing.
\newblock \emph{Genome Research}, 6\penalty0 (5):\penalty0 448--453, 1996.
\newblock \doi{10.1101/gr.6.5.448}.
\newblock URL \url{http://genome.cshlp.org/content/6/5/448.abstract}.

\end{thebibliography}

\end{document}